\documentclass[12pt,en,authoryear]{elegantpaper}

\usepackage[boxed]{algorithm2e}
\SetAlCapSkip{3pt}
\usepackage{caption}
\usepackage{subcaption}
\usepackage{wrapfig}
\usepackage{nicefrac}
\usepackage{mathtools}
\mathtoolsset{showonlyrefs=true}
\usepackage{relsize}
\usepackage{thmtools}
\usepackage{thm-restate}

\DeclareMathOperator*{\maximize}{maximize}

\DeclareMathOperator*{\argmin}{arg\,min}

\newcommand{\bsucceq}{\boldsymbol{\succeq}}

\newcommand{\eps}{\varepsilon}
\newcommand{\bbR}{\mathbb{R}}
\newcommand{\bbN}{\mathbb{N}}

\newcommand{\cR}{\mathcal{R}}

\newcommand{\cD}{\mathcal{D}}

\newcommand{\cC}{\mathcal{C}}

\newcommand{\bb}{{\mathbf{b}}}
\newcommand{\bx}{{\mathbf{x}}}
\newcommand{\by}{{\mathbf{y}}}

\newcommand{\bh}{{\mathbf{h}}}

\title{Constrained Serial Rule on the Full Preference Domain}
\author{Priyanka Shende \footnote{University of California, Berkeley; Email: priyanka.s@berkeley.edu}  }

\begin{document}

\maketitle

\begin{abstract}
We study the problem of assigning objects to agents in the presence of arbitrary linear constraints when agents are allowed to be indifferent between objects. Our main contribution is the generalization of the (Extended) Probabilistic Serial mechanism via a new mechanism called the Constrained Serial Rule. This mechanism is computationally efficient and maintains desirable efficiency and fairness properties namely constrained ordinal efficiency and envy-freeness among agents of the same type. Our mechanism is based on a linear programming approach that accounts for all constraints and provides a re-interpretation of the ``bottleneck'' set of agents that form a crucial part of the Extended Probabilistic Serial mechanism.

\end{abstract}

\newpage

\section{Introduction}
\label{sec:intro}
Allocating a number of indivisible objects among a set of agents in a \emph{fair} and \emph{efficient} manner is one of the most fundamental problems in applied and theoretical economics. Classical models for object allocation without monetary transfers such as the house allocation model (\cite{hylland1979efficient}) are now well understood and mechanisms that guarantee any maximal subset of attainable properties along the dimensions of \emph{fairness}, \emph{efficiency}, and \emph{incentive compatibility} are well known. However, many applications in practice impose additional constraints on the set of allowed allocations and designing allocation mechanisms in such constrained settings remains an ongoing challenge. In this work, we present a novel \emph{Constrained Serial Rule} mechanism that always obtains efficient and fair outcomes for object allocation under a large class of general constraints.

In the traditional object allocation model, a finite set of indivisible objects must be allocated to a set of agents based on the agents' preferences over those objects. In this paper, we restrict our attention to \emph{ordinal} mechanisms where the agents report only a preference ranking over objects\footnote{This is in contrast with \emph{cardinal} mechanisms where the agents report cardinal utilities for each object.}. This setting models a large number of real-world applications such as placement of students to public schools \citep{abdulkadirouglu2003school}, course allocation \citep{budish2011combinatorial}, organ donation \citep{roth2005pairwise}, and on-campus housing allocation \citep{chen2002improving} among others. In many of these applications, it is highly desirable that the mechanism is \emph{fair}, i.e. no agent is discriminated against and also \emph{efficient}, i.e. there is no other outcome that is preferred by all agents. Unfortunately, since objects are indivisible, it can be easily observed that no mechanism can be perceived as being fair ex-post. Randomization is, therefore,  often used as a tool to restore fairness from an ex-ante perspective. 

In the context of house allocation, \citet{bogomolnaia2001new} introduced the notion of \emph{ordinal efficiency}. A random assignment is said to be ordinally efficient if there exists no other assignment that stochastically dominates it. Indeed, they showed that the well-studied \emph{Random Priority} mechanism that orders agents in a uniformly random order and then allocates each agent her most preferred object from the set of remaining objects is not ordinally efficient and only satisfies a weaker notion of ex-post efficiency. 
The notion of \emph{envy-freeness} that requires each agent to prefer her own allocation to anyone else's allocation is often considered as the gold standard of fairness in many different settings such as resource allocation~\citep{foley1967resource}, cake-cutting~\citep{robertson1998cake}, and rent division~\citep{edward1999rental}. In the context of random assignment mechanisms, this property is formulated using the first-order stochastic dominance relation. In their seminal work, \cite{bogomolnaia2001new} introduced the \emph{Probabilistic Serial} mechanism that always produces ordinally efficient and envy-free outcomes. The probabilistic serial mechanism can be described as follows. At time zero, each agent begins ``eating'' her most preferred object. An object becomes unavailable once the total time spend by agents eating it equals one. Once an object becomes unavailable, all agents that were eating it, switch to eating their most preferred object among the ones still available. Finally, the probability that an agent receives some object is the time spent by that agent to eat the object.

Unfortunately, the probabilistic serial mechanism assumes that agents have strict preferences over objects, which is a fairly restrictive assumption in practice. Indeed, as discussed by \citet{katta2006solution} and \citet{erdil2017two}, ties in preferences are widespread in many practical applications. For example, agents may treat some objects as identical. Even when objects are all distinct, evaluating and ranking all objects may be computationally prohibitive, and agents may only reveal coarse rankings with indifferences. In their influential paper, \citet{katta2006solution} present the \emph{Extended Probabilistic Serial} mechanism that generalizes the probabilistic serial mechanism to the full preference domain and retains the desirable properties of ordinal efficiency and envy-freeness.

Widespread applicability of these mechanisms is hindered by the fact that these mechanisms assume that every random assignment is feasible. In a large number of practical applications, legal and policy requirements necessitate studying mechanisms where the set of feasible assignments is constrained in some way. For example, in the course allocation problem, there are often requirements on the minimum (and maximum) number of students assigned to a course. Similarly, in school choice applications, it is required to find assignments that maintain a minimum level of diversity (\cite{ehlers2014school}). In resident matching, it is often necessary for the allocation of doctors to hospitals to satisfy geographic constraints (\cite{kamada2015efficient}). In the refugee resettlement problem, objects represent settlement facilities and a feasible assignment must be such that the total demands of all agents assigned to a facility must be met by the total supply of resources at that facility (\cite{delacretaz2016refugee}). Similarly, in kidney matching applications (\cite{roth2005pairwise}), blood-type compatibility imposes constraints on feasible matchings. 
In this paper, we study the object allocation problem with arbitrary linear constraints on the set of feasible probabilistic assignments. Following the work of \citet{balbuzanov2019constrained}, this formalization supports an arbitrary set of constraints on ex-post allocations.

\subsection{Our Contributions}
\label{sec:contributions}
Our primary contribution is to generalize the probabilistic serial mechanism to the full preference domain and support arbitrary linear constraints on the set of feasible random assignments via a new mechanism called the \emph{Constrained Serial Rule}. Our mechanism is computationally efficient and only requires a running time that is polynomial in the number of constraints, agents, and objects.
For the classical unconstrained house allocation setup on the full preference domain, our mechanism coincides with the extended probabilistic serial mechanism of \citet{katta2006solution}. Various generalizations of the probabilistic serial mechanism have been proposed for multi-unit demand (\cite{kojima2009random}), specific type of constraints such as bi-hierarchical constraints \citep{budish2013designing}, type-dependent distributional constraints \citep{ashlagi2020assignment},  combinatorial demand \citep{nguyen2016assignment}, property rights with individual rationality \citep{yilmaz2010probabilistic}, and even arbitrary constraints on the ex-post allocations \citep{balbuzanov2019constrained}. Our \emph{constrained serial rule} unifies this literature and provides a common generalization of all these mechanisms and also provides an extension to the full preference domain.

We show that the \emph{constrained serial rule} maintains the desirable efficiency and fairness properties of the probabilistic serial mechanism even in our general constrained setting. In particular, our mechanism is constrained ordinally efficient. While it is easy to observe that arbitrary constraints rule out the existence of envy-free mechanisms, we show that the \emph{constrained serial rule} maintains a compelling notion of fairness. Intuitively, we say agents $i$ and $j$ are of the same \emph{type} if the constraint structure does not distinguish between the two agents. We show that the constrained serial rule mechanism guarantees envy-freeness among any pair of agents of the same type. However, our mechanism is not strategyproof or even weak-strategyproof. This is unsurprising since even in the unconstrained setting, weak-strategyproofness is incompatible with ordinal efficiency and envy-freeness on the full preference domain \citep{katta2006solution}.

\subsection{Other Related Work}
\label{sec:related}
There is a growing body of literature on assignment and matching mechanisms subject to constraints. Several studies have considered floor and ceiling constraints in the context of controlled school choice, college admissions, and affirmative action  \citep{kojima2012school,kominers2013designing,hafalir2013effective,hamada2016hospitals,fragiadakis2017improving,fleiner2016matroid,goto2015improving,westkamp2013analysis,ehlers2014school,echenique2015control,biro2010college,ashlagi2020assignment}. \cite{echenique2019constrained} consider arbitrary ex-post constraints as in \cite{balbuzanov2019constrained} and provide a pseudo-market equilibrium solution that is constrained ex-ante efficient and fair.

\cite{budish2013designing,pycia2015decomposing,akbarpour2014approximate} have studied the implementability of random matching mechanisms. \cite{budish2013designing} identify bi-hierarchical constraint structures as a necessary and sufficient condition for implementing a random assignment using lottery of feasible assignments. They also provide a generalization of the (extended) probabilistic serial mechanism in the case when there are no floor constraints. Indeed, our mechanism is able to accommodate bi-hierarchical constraint inequalities in the presence of non-zero floor constraints. \cite{akbarpour2014approximate} consider more general constraints beyond bi-hierarchical structures and show how feasible random assignments can be implemented approximately. \cite{pycia2015decomposing} provide sufficient conditions on the properties of random mechanisms that continue to be satisfied on the deterministic mechanisms when random mechanisms are decomposed as a lottery over these deterministic mechanisms. While the focus of our paper is not on implementability, we provide a small discussion of this in Section \ref{sec:implementability}.

\section{Model and Preliminaries}
We consider a finite set $N$ of agents and a finite set $O$ of objects. Let $n=|N|$ be the number of distinct agents and $\rho = |O|$ be the number of distinct objects. Every agent has a unit\footnote{Our results also generalize to the case when all agents demand $d \geq 1$ objects.} demand and each object $o \in O$ is supplied in $q_o \in \bbN$ copies. When objects are scarce, we can include the null object, $\emptyset$, in the set $O$, which is supplied in a quantity sufficient to meet the demand of all agents. That is, $q_{\emptyset} \geq |N|$. We can, therefore, without loss of generality, assume that $\sum_{o \in O} q_o \geq n$. Each agent $i \in N$ has a preference relation $\succeq_i$ on the set of objects in $O$. The preference $\succeq_i$ is assumed to be complete and transitive. In particular, we allow agents to be indifferent between any pair of objects in $O$. Let $E(\succeq)$ be the number of indifference classes within the preference $\succeq$. For any $\ell \leq E(\succeq_i)$, let $T_i(\ell)$ be the set of objects in the first $\ell$ indifference classes of the preference $\succeq_i$.
A set of individual preferences of all agents constitutes a preference profile $\bsucceq=(\succeq_i)_{i \in N}$. Let $\cR$ denote the set of all complete and transitive relations on $O$ and $\mathcal{R}^n$ be the set of all possible preference profiles.

A \emph{random assignment} of objects to agents is given by a vector $\bx =(x_{i,o})_{i \in N, o \in O} \in [0,1]^{n\rho}$ such that
\begin{align*}
    &\sum_{o \in O}x_{i,o} = 1 \quad\ \  \forall i \in N \\
    &\sum_{i \in N} x_{i,o} \leq q_o \quad \forall o \in O
\end{align*}
 In assignment $\bx$, every agent $i$'s \emph{allocation} is given by the sub-vector $\bx_i =(x_{i,o})_{o \in O}$, where the quantity $x_{i,o}$ is interpreted to be the probability with which object $o$ is assigned to agent $i$. Let $x_i(S) = \sum_{o \in S}x_{i,o}$ be agent $i$'s cumulative allocation for the set of objects in set $S$. An assignment is \emph{deterministic} whenever $x_{i,o} \in \{0,1\}$, i.e, every agent is assigned a single object with probability $1$. Let $\cD$ denote the set of all deterministic assignments and $\Delta \cD$ denote the set of all random assignments. 
A random assignment \emph{mechanism} is a mapping, $ \varphi:\mathcal{R}^n \rightarrow \Delta\mathcal{D}$, that associates each preference profile $\bsucceq \in \mathcal{R}^n$ with some random assignment $\bx \in \Delta \cD$.

We extend agents' preferences from the set of objects to the set of random allocations using the stochastic dominance relation. Given two random assignments $\bx$ and $\by$, allocation $\bx_i$ \emph{stochastically dominates} allocation $\by_i$ with respect to $\succeq_i$, denoted by $\bx_i sd(\succeq_i) \by_i$, if and only if $\sum_{o' \succeq_i o} x_{i,o'} \geq \sum_{o' \succeq_i o} y_{i,o'}$ for all $o \in O$. If the inequality is strict for some $o \in O$, then $\bx_i$ \emph{strictly} stochastically dominates $\by_i$, in which case we denote it by $\bx_i sd(\succ_i) \by_i$.

We now introduce a general class of constraints into our model. At any given preference profile, we assume that the set of feasible random assignments can be described as a convex polytope. Formally, at preference profile $\bsucceq$, the set of feasible random assignments $\Delta\cC(\bsucceq)$ is parameterized by a matrix $A = [a^c_{i,o}]_{1\leq c\leq m,\{i,o\} \in N \times O} \in \bbR^{m \times n\rho}$ and a vector $\boldsymbol{b}=[b^c]_{1\leq c\leq m} \in \bbR^m$, where $c$ is a generic constraint and $m$ is the number of constraints, and is defined as:
\begin{align*}
    \Delta\mathcal{C}(\bsucceq) = \{\bx \in \Delta \cD \mid A\bx \leq \boldsymbol{b} \}
\end{align*}

We assume that at each preference profile $\bsucceq$, the set of feasible random assignments is non-empty. That is, $\Delta \cC(\bsucceq) \neq \emptyset$. Such a formulation of the constraints enables us to apply our model to many different specific applications that we describe in Section \ref{sec:applications}. Let $\Delta \cC = \{\Delta C(\bsucceq)\}_{\bsucceq \in \cR^n}$ be the collection of constraint polytopes for all preference profiles. Given a collection of constraints $\Delta \cC$, a mechanism is \emph{feasible} if at every preference profile $\bsucceq$, $\varphi(\bsucceq) \in \Delta \cC(\bsucceq)$.

We next define the normative properties of efficiency and fairness in the presence of constraints. 

\begin{definition}[Constrained Ordinal Efficiency]
\label{def:ordinaleff}
A random assignment $\bx$ is \emph{constrained ordinally efficient} at a preference profile $\bsucceq$ and constraint set $\Delta \cC(\bsucceq)$ if there does not exist another random assignment $\bx' \in \Delta \cC(\bsucceq)$ such that $\bx'_isd(\succeq_i)\bx_i $ for all $i \in N$, with $\bx'_isd(\succ_i)\bx_i$ for  at least one $i \in N$. A mechanism $\varphi$ is \emph{constrained ordinally efficient} if for every preference profile $\bsucceq$, $\varphi(\bsucceq)$ is constrained ordinally efficient.
\end{definition}

The classic notion of fairness requires that no agent should envy the allocation received by some other agent. When faced with arbitrary feasibility constraints, it is easy to see that one cannot guarantee the existence of envy-free assignments. Therefore, we restrict ourselves to fairness comparisons between agents that belong to the same type. For a given constraint matrix $A$, we say agents $i$ and $j$ belong to the same type if for every object $o$, the variables $x_{i,o}$ and $x_{j,o}$ have the same coefficients in every constraint in $A$.

\begin{definition}[Agent Type]
\label{def:sametype}
Let $\Delta\mathcal{C}(\bsucceq) = \{\bx \in \Delta \cD \mid A\bx \leq \boldsymbol{b} \}$  where $A = [a^c_{i,o}]_{1\leq c\leq m,\{i,o\} \in N \times O}$ denote the constraint set a preference profile $\bsucceq$.
Two agents $i$ and $j$ are said to be of the \emph{same type} at this profile, if for every constraint $1 \leq c \leq m$, and for every object $o \in O$, we have $a^c_{i,o}=a^c_{j,o}$. 
\end{definition}

\begin{definition}[Envy-freeness among agents of the same type]
\label{def:envy-free}
A random assignment $\bx$ is \emph{envy-free among agents of the same type} if for every pair of agents $i,j \in N$ of the same type, $\bx_i sd(\succeq_i) \bx_j$. A mechanism $\varphi$ is \emph{envy-free for agents of the same type} if for every preference profile $\bsucceq$, $\varphi(\bsucceq)$ is envy-free among agents of the same type.
\end{definition}

\section{The Constrained Serial Rule}

We first give a brief intuitive description of the classical probabilistic serial mechanism~\citep{bogomolnaia2001new} in the simple house allocation model\footnote{Agents have strict preferences over objects and there are no additional constraints on the assignment.}.
The mechanism is best described as a continuous time procedure for $t\in [0,1]$: at each infinitesimal time interval $[t, t+dt)$, each agent $i$ consumes $dt$ amount of her most preferred object among the set of objects currently available. When this procedure terminates, the probability that an agent is assigned an object is given by the fraction of the object consumed by the agent. 

While attempting to extend the probabilistic serial mechanism to our general model on the full preference domain and with arbitrary constraints on the eventual random assignment, one faces two key challenges. First, when agents have strict preferences, each agent at any point in time has a unique most preferred object and hence the mechanism can simply allocate that object to the agent. On the other hand, when agents are indifferent between two or more objects, the mechanism can no longer uniquely identify an object to assign. \citet{katta2006solution} deal with this difficulty by constructing a flow graph where every agent points to her set of most preferred objects and then using a parametric flow formulation to find a set of \emph{bottleneck} agents and objects. Intuitively, the set of bottleneck agents are those that compete the most among themselves and the bottleneck objects are those desired by bottleneck agents. Once the bottleneck agents have been identified, the extended probabilistic serial mechanism allocates all bottleneck objects among these agents uniformly. As in the classic probabilistic serial rule, the mechanism then proceeds by each bottleneck agent simply pointing to her next most preferred object. A key observation is that the (extended) probabilistic serial mechanism attempts to assign each agent her most preferred object for as long as possible. In fact, as observed by \citet{bogomolnaia2015random}, one can provide a welfarist interpretation of the probabilistic serial rule as follows: for any agent $i$, let $x_i(\ell)$ be the total probability share of objects that agent $i$ receives for her top $\ell$ indifference classes; then the probabilistic serial mechanism leximin maximizes the vector of all such shares $(x_i(\ell))_{i \in N, k \leq E(\succeq_i)}$. We crucially use this observation in our \emph{constrained serial rule} mechanism.

The second challenge arises due to the presence of arbitrary constraints on the space of feasible random assignments. We observe that at any step of the mechanism, the (extended) probabilistic serial mechanism always assigns each agent her most preferred object (at that time). In the presence of constraints, however, it is essential to not allow agents to obtain their most preferred object if such an allocation leads to infeasibility. More precisely, the mechanism needs to \emph{look ahead} in time as it 
builds up a partial allocation to ensure that there is at least one way to extend the partial allocation to a feasible assignment. Our \emph{constrained serial rule} mechanism uses a linear program that explicitly accounts for all constraints and at every step maintains a feasible solution.

We now describe the linear program that is used crucially by the mechanism. Let $S \subseteq N$ denote any subset of agents and let $\ell_i \in \{1, 2, \ldots, E(\succeq_i)\}$ for each agent $i \in N$ denote an indifference threshold. Let $F$ be a set of triples that denotes prior promised assignments. Formally, a triple $(i, \ell, \gamma) \in F$ indicates that agent $i$ must receive a total probability share of at least $\gamma$ from her top $\ell$ indifference classes.
The linear program $LP(S, F, (\ell_i)_{i \in N})$ specified in Figure \ref{fig:lp} finds a random assignment $\bx \in \bbR_+^{n\rho}$ that satisfies all constraints specified by $F$ in addition to the imposed feasibility constraints and maximizes the total probability share that each agent $i \in S$ receives from her top $\ell_i$ indifference classes. The variables $h_i$ for each $i \in N$ represent the total probability share received by agent $i$ for objects in her top $\ell_i$ indifference classes. Constraints \eqref{eq:lp_const1} and \eqref{eq:lp_const2} enforce the requirement that the linear program maximizes $\min_{i \in S} h_i$. Constraints \eqref{eq:lp_const3} and \eqref{eq:lp_const4} enforce that the obtained random assignment is feasible, and finally constraint \eqref{eq:lp_const5} requires the assignment to be consistent with the requirements specified by the triples in $F$.

\begin{figure}
    \centering
    \begin{align}
    LP(S,F,(\ell_i)_{i \in S}) = &\quad \maximize_{\bx, \bh, \lambda}          & \lambda  &      &  \label{eq:lp_obj}\\
    &\quad \mathrm{s.t.} &   h_i &\geq \lambda &  \forall i \in S  \label{eq:lp_const1}\\
    &  &  \sum_{o \in T_i(\ell_i)} x_{i,o}  & \geq h_i & \forall i \in N \label{eq:lp_const2}\\
                &  &  A\bx   & \leq \boldsymbol{b}  &   \label{eq:lp_const3}\\
                &  & \bx &\in \Delta \mathcal{D}& \label{eq:lp_const4}\\
                &  &   \sum_{o \in T_i(\ell)} x_{i,o}  &\geq \gamma & \forall (i,\ell,\gamma) \in F \label{eq:lp_const5}\\
                &  & \bh, \lambda & \geq 0 &
\end{align}
    \caption{Linear Program used by the Constrained Serial Rule}
    \label{fig:lp}
\end{figure}

\subsection{Mechanism}

We are now ready to describe the constrained serial rule mechanism formally.
The mechanism proceeds in multiple rounds. We initialize $F^1 = \emptyset$ and for each agent $i \in N$, we initialize $\ell_i^1 = 1$. Intuitively, $\ell_i^t$ denotes the threshold indifference class for agent $i$ in round $t$. In other words, in round $t$, we consider the total probability share of objects in $T_i(\ell_i^t)$ assigned to agent $i$. Let $h_i^t$ denote the total probability share of objects in the top $\ell_i^t$ indifference classes assigned to agent $i$.
We use the linear program described in Figure \ref{fig:lp} to find a feasible random assignment such that $\min_{i} h_i^t$ is maximized. The mechanism then identifies a set $B^t$ of \emph{bottleneck} agents. Intuitively, these are the set of agents who are responsible for the $LP$ objective in this round to be only $\lambda^t$. Since we are dealing with arbitrary linear constraints, our definition of \emph{bottleneck} agents needs to be more subtle than that of \citet{katta2006solution}.
We define $B_t$ to be a \emph{minimal} set of agents such that solving the linear program while only attempting to maximize the utility of agents in that set also yields the same objective value of $\lambda^t$. Our definition of bottleneck agents is central to the validity of the mechanism as well as its efficiency and fairness properties. Finally, once the bottleneck set of agents has been identified, we update $F^t$ to guarantee that in future rounds each agent $i \in B^t$ obtains at least the promised $\lambda^t$ probability share from her top $\ell_i^t$ indifference classes and then increment the threshold $\ell_i^t$ for all such agents. The mechanism then proceeds to the next round and the process continues until every agent receives a total probability share of 1. 
Algorithm \ref{alg:genps} provides a complete formal description of the algorithm.

\begin{algorithm}[htb]
{\bf Initialize:}\\
$t \leftarrow 1$\;
$\ell_i^t \leftarrow 1,\quad \forall i \in N$\;
$F^t \leftarrow \emptyset$\;
\For{$t = 1, 2, \ldots$}
{
$(\bx^t, \bh^t, \lambda^t) \leftarrow LP(N, F^t, (\ell_i^t)_{i \in N})$\;
\eIf{$\lambda^t = 1$}
{
$\bx \leftarrow \bx^t$\;
terminate\;
}
{
Find a minimal set $B^t$ such that $LP(B^t, F^t, (\ell_i^t)_{i \in N})$ has objective value $\lambda^t$\;
Update $F^{t+1} = F^t \cup \{(i, \ell_i^t, \lambda^t) \mid i \in B^t\}$\;
Update $\ell_i^{t+1} = \begin{cases}
                    \ell_i^t + 1 \quad \forall i \in B^t\\
                    \ell_i^t \quad\quad\quad \text{otherwise}
                    \end{cases}$
}
}
\caption{The Constrained Serial Rule}
\label{alg:genps}
\end{algorithm}

We provide a simple example run of Algorithm \ref{alg:genps} on a constrained allocation problem to illustrate our constrained serial rule.
\begin{example}
Consider an allocation problem with three agents $N = \{1,2,3\}$ and three objects $O = \{a,b,c\}$. Our goal is to obtain an assignment that satisfies the usual bistochastic constraints, i.e. $\sum_{o \in O} x_{i,a} = 1\ \forall i \in N$ and $\sum_{i \in N} x_{i,a} = 1\ \forall o \in O$. In addition, there are two additional constraints as follows: $x_{1,a} + x_{2,a} \leq 0.5$ and $x_{1,c} + x_{2,c} \geq 0.5$.
The agents' preferences are given as follows.
\begin{align*}
    &\succ_1: \quad\quad a \quad\quad b \quad\quad c\\
    &\succ_2: \quad\quad \{a,b\} \quad\quad\ \  c\\
    &\succ_3: \quad\quad c \quad\quad b \quad\quad a
\end{align*}

We illustrate how our mechanism works through this example. In the first round, we initialize $\ell_1^1 = \ell_2^1 = \ell_3^1 = 1$ and solve the linear program to find a feasible solution that maximizes $\lambda^1 = \min\{x_{1,a}, x_{2,a} + x_{2,b}, x_{3,c}\}$. With the given constraints, one potential optimum solution assigns $x_{1,a} = x_{2,b} = x_{3,c} = 0.5$ to obtain $\lambda^1 = 0.5$. We then proceed to find the set of bottleneck agents. In this example, either of the singleton sets with agents $1$ or $3$ could be the bottleneck set. This is because the constraint $x_{1,a} + x_{2,a} \leq 0.5$ prevents agent $1$ from receiving a larger probability share of object $a$. Similarly, the constraint $x_{1,c} + x_{2,c} \geq 0.5$ prevents agent $3$ from receiving a larger probability share of object $c$. Suppose we select agent $1$ to be the bottleneck agent. Then, we increment $\ell_1^2 = 2$ and maintain $\ell_2^2 = \ell_3^2 = 1$. We also set $F^2 = (1,1,0.5)$ to signify that agent $1$ must continue to obtain $0.5$ amount of her top choice.

In the second round, we again solve the linear program to find a feasible solution that now maximizes $\lambda^2 = \min\{x_{1,a} + x_{1,b}, x_{2,a} + x_{2,b}, x_{3,c}\}$. As described earlier, agent $3$ is unable to receive more than $0.5$ amount of object $c$ in any feasible solution and hence we obtain $\lambda^2 = 0.5$. In this case, agent $3$ is the unique bottleneck agent and as earlier we increment her indifference threshold and add the triple $(3,1,0.5)$ to $F^3$.
Similarly, in the third round, we solve the linear program to find a feasible solution that maximizes $\lambda^3 = \min\{x_{1,a} + x_{1,b}, x_{2,a} + x_{2,b}, x_{3,c} + x_{3,b}\}$. However, the constraints $x_{1,a} + x_{2,a} \leq 0.5$ and $x_{1,a} + x_{2,a} + x_{3,a} = 1$ together imply that $x_{3,a} \geq 0.5$ and thus $x_{3,c} + x_{3,b} \leq 0.5$. So yet again, we have $\lambda^3 = 0.5$ and agent $3$ is the bottleneck agent and we increment her indifference threshold and add the triple $(3,2,0.5)$ to $F^4$.

In the fourth round, the linear program attempts to find a feasible solution that respects all the constraints in $F^4$ and maximizes $\lambda^4 = \min\{x_{1,a} + x_{1,b}, x_{2,a} + x_{2,b}, x_{3,c} + x_{3,b} + x_{3,a}\}$. In this case, one potential optimum solution assigns $x_{1,a} = 0.5, x_{1,b} = 0.25, x_{2,b} = 0.75, x_{3,c} = 0.5, x_{3,a} = 0.25$ to obtain $\lambda^4 = 0.75$. Since the constraint $x_{1,a} + x_{2,a} = 0.5$ is already tight and object $b$ has been fully allocated, both agents $1$ and $2$ are in the bottleneck set in this round. We increment the indifference threshold of both the agents to obtain $\ell_1^5 = 3$ and $\ell_2^5 = 2$ and add the triples $(1,2,0.75)$ and $(2,1,0.75)$ to $F^5$.

Finally, in the fifth round, we again solve the linear program to find a feasible solution that respects all the constraints in $F^5$ and  maximizes $\lambda^5 = \min_{i \in N}\{x_{i,a} + x_{i,b} + x_{i,c}\}$. Since any feasible solution satisfies $x_{i,a} + x_{i,b} + x_{i,c} = 1$, we obtain $\lambda^5 = 1$ and the mechanism terminates. The outcome of the mechanism is any feasible solution that satisfies all constraints in $F^5$. For this example, the unique such solution is given by the following random assignment.
\begin{table}[h]
    \centering
    \begin{tabular}{llll}
 & a & b & c \\
\bottomrule
1 & 0.5 & 0.25 & 0.25 \\
2 & 0 & 0.75 & 0.25 \\
3 & 0.5 & 0 & 0.5 \\
        \bottomrule
       \end{tabular}
    \label{tab:my_label}
\end{table}
\end{example}

\subsection{Properties}

We first show that the constrained serial rule presented in Algorithm \ref{alg:genps} is well-defined and always produces a feasible random assignment.

\begin{proposition}
\label{prop:terminates}
Algorithm \ref{alg:genps} terminates and always produces a feasible random assignment.
\end{proposition}

\begin{proof}
We first observe that in any step $t$, the $LP$ always has a feasible solution. This is because, for any $t > 1$, the solution $(x^{t-1}, h^{t-1}, \lambda^{t-1})$ from the previous iteration continues to be a feasible solution; whereas for $t=1$, the existence of a feasible solution is guaranteed since the assignment constraints are assumed to be satisfiable.

In any step $t$, it is easy to see that the bottleneck set of agents $B^t$ is guaranteed to exist by observing that the set of all agents $N$ satisfies the required constraint by definition. Further, we observe that any agent $i^* \in N$ such that $\ell_{i^*}^t = E(\succeq_{i^*})$ does not appear in the bottleneck set $B^t$. This is because, for any such agent $i^* \in N$, the linear programs $LP(B^t, F^t, (\ell_i^t)_{i \in N})$ and $LP(B^t \setminus \{i^*\}, F^t, (\ell_i^t)_{i \in N})$ are actually identical and by definition $B^t$ is the \emph{minimal} set that satisfies the condition.
Thus in any step $t$, we increment the indifference threshold of at least one agent. Finally, by definition of feasible assignment, when $\ell_i^t = E(\succeq_i),\ \forall i \in N$, we may set $h^t_i = \sum_{o \in O} x^t_{i,o} = 1$ for all $i \in N$ and hence obtain $\lambda^t = 1$ and the algorithm terminates. Since any agent $i \in N$ has $E(\succeq_i) \leq \rho$, the mechanism terminates in at most $n\rho$ rounds.

Finally, since the linear program in Figure \ref{fig:lp} explicitly maintains the constraints $A\bx \leq \boldsymbol{b}$ and $x \in \Delta \cD$, the outcome is guaranteed to be feasible.
\end{proof}

We then proceed to show two technical lemmas. The first lemma shows that even though we allow arbitrary linear constraints on the random assignment, the projection of the feasibility polytope on the $\bh$ and $\lambda$ variables satisfies a desirable monotonicity property. Specifically, the constraints on the $\bh$ variables can be expressed as a set of linear upper-bound constraints with only non-negative coefficients.

\begin{lemma}
\label{lem:polytope}
Fix any step $t$ of the algorithm. Let $P$ be the polytope defined by the constraints in $LP(B^t, F^t, (\ell_i^t)_{i \in N})$ and let $Q = \{(\bh, \lambda) \mid \exists (\bx, \bh, \lambda) \in P\}$ be its projection. Then there exists a non-negative matrix $\tilde A$ and a non-negative vector $\tilde b$ such that
\[Q = \begin{cases}
h_i \geq \lambda, \quad \forall i \in B^t\\
\tilde A \bh \leq  \boldsymbol{\tilde{b}} \\
\bh, \lambda \geq 0
\end{cases}\]
\end{lemma}
\begin{proof}
By the Fourier-Motzkin theorem, we know that $Q$ is also a polytope that can be obtained from $P$ using Fourier-Motzkin elimination. Since the only constraints involving $\lambda$ in $P$ do not contain any $\bx$ variables, those constraints remain unchanged in polytope $Q$. Let $\tilde A = [\tilde a^c_{i}]$ and $\tilde b = [\tilde b^c]$ denote the minimal linear constraints on the variables $\bh$ obtained after Fourier-Motzkin elimination. We now need to show that $\tilde A$ and $\tilde b$ are non-negative.

We first observe that the polytope $Q$ is downward closed on the $\bh$ variables, i.e. for any $\bh' \leq \bh$, if $(\bh, \lambda) \in Q$ then there exists a $\lambda' \leq \lambda$ such that $(\bh', \lambda') \in Q$. This is because, by definition, if $(\bh, \lambda) \in Q$ then there exists an $\bx$ such that $(\bx, \bh, \lambda) \in P$. Further, observing the polytope $P$ (refer to Figure \ref{fig:lp}), it is clear that $(\bx, \bh', \lambda')$ also belongs to $P$ where $\lambda' = \min_{i \in B^t} h'_i$, and thus $(\bh', \lambda') \in Q$.

Now, suppose for contradiction that there exists a constraint $c$ such that the entry $\tilde a^c_j < 0$ for some $j \in N$. Since this constraint is not redundant, there exists a vector $\tilde \bh \in \bbR_+^{n}$ such that the $c$th constraint is the only binding constraint, i.e., $\sum_{i} \tilde a^c_i \tilde h_i = \tilde b^c$ and further $\tilde h_i > 0$ for all $i \in N$. Define $\bh' \in \bbR_+^{n}$ as $h'_i = \tilde h_i,\ \forall i \in N \setminus \{j\}$ and $h'_j = 0$. By definition $\bh' < \tilde \bh$, and yet $\sum_{i} \tilde a^c_i h'_i > \tilde b^c$ and thus $\bh' \notin Q$ which is a contradiction.

Finally if the matrix $\tilde A$ is non-negative and the vector $\tilde b$ has a negative entry, then the polytope $Q$ must be empty. Since $Q$ is guaranteed to be non-empty, we must have that $\tilde b$ is non-negative.
\end{proof}

The following lemma demonstrates the importance of our definition of \emph{bottleneck} agents. Informally, it states that if $B^t$ is the set of bottleneck agents at some step $t$ of the algorithm, then no agent in $B^t$ can obtain a higher total allocation for her top $\ell_i^t$ indifference classes without hurting some other agent in $B^t$. This lemma is crucial to proving the constrained serial rule produces constrained ordinally efficient outcomes, and also to demonstrate its fairness properties.

\begin{lemma}
\label{lem:bottleneck}
Fix any step $t$ of the algorithm and let $(\bx^t, \bh^t, \lambda^t)$ denote an optimal solution to $LP(B^t, F^t, (\ell_i^t)_{i \in N})$. Then
  $\nexists\ \by = (y_{i,o})_{i \in N, o \in O} \in \mathbb{R}^{n\rho}_+$ such that 
 \begin{enumerate}[nosep]
    \item $A \by \leq b$, $\by \in \Delta \cD$
    \item $y_i(T_i(\ell)) \geq \gamma,\ \forall (i,\ell,\gamma) \in F^t$
    \item $y_i(T_i(\ell^t_i)) \geq \lambda^t,\ \forall i \in B^t$ with atleast one strict inequality.
\end{enumerate}
\end{lemma}

\begin{proof}
Let $P$ be the polytope defined by the constraints in $LP(B^t, F^t, (\ell_i^t)_{i \in N})$. Let $Q = \{(\bh, \lambda) \mid \exists (\bx, \bh, \lambda) \in P)\}$ be the projection of $P$ on the $\bh$ and $\lambda$ variables. By Lemma \ref{lem:polytope}, there exists a non-negative matrix $\tilde A = [\tilde a^c_i]$ and a non-negative vector $\tilde b$ such that
\[Q = \begin{cases}
h_i \geq \lambda, \quad \forall i \in B^t\\
\tilde A \bh \leq \boldsymbol{\tilde{b}} \\
\bh, \lambda \geq 0
\end{cases}\]

By definition, $(\bx^t, \bh^t, \lambda^t)$ is an optimal solution to $LP(B^t, F^t, (\ell_i^t)_{i \in N})$. Let $\tilde \bh$ be defined as $\tilde{h}_i = \lambda^t$ for all $i \in B^t$ and $\tilde{h}_i = 0$ for all $i \notin B^t$. We now have $(\bx^t, \tilde \bh, \lambda^t)$ is also an optimal solution to 
$LP(B^t, F^t, (\ell_i^t)_{i \in N})$. Thus, $(\tilde \bh, \lambda^t)$ does not lie in the interior of the polytope $Q$. Thus at least one of the constraints $\tilde A \bh \leq \tilde b$ must be tight at this point, i.e., there exists a constraint $c$ such that  $\sum_{i \in B^t} \tilde a^{c}_{i} \tilde{h}_i = \tilde b^c$. 
We now consider two cases.

{\emph {Case 1}: $\tilde a^c_{i} > 0$ for all $i \in B^t$}.
 Suppose for contradiction that there exists a $\by$ that satisfies the premises of the lemma. For any agent $i \in N$, let $h'_i =  y_{i}(T_i(\ell_i^t))$, so we have $h'_i \geq \lambda^t$ for all $i \in B^t$ with at least one strict inequality. Thus since $(\by, \bh',\lambda^t)$ is a feasible solution to 
 $LP(B^t, F^t, (\ell_i^t)_{i \in N})$, we have $(\by, \bh',\lambda^t) \in P$. By definition of $Q$, we have $(\bh', \lambda^t) \in Q$. However, we have
\begin{align*}
    \sum_{i \in B^t} \tilde a^c_{i}h'_i > \sum_{i \in B^t} \tilde a^c_{i}\lambda^t = \sum_{i \in B^t} a^c_{i}\tilde{h}_i = \tilde b^c
\end{align*}
which is a contradiction to the statement that $(\bh', \lambda^t) \in Q$. 

{\emph {Case 2}: $\tilde a^c_{j} = 0$ for some $j \in B^t$}. Consider the set $B' = B^t \setminus \{j\}$. We now have $\sum_{i \in B'} \tilde a^c_i \tilde h_i = \tilde b^c$. 
Let $P'$ be the polytope defined by the constraints in $LP(B', F^t, (\ell_i^t)_{i \in N})$ and $Q'$ be its projection.
Since $B^t$ is the minimal bottleneck set, there exists a $(\by',h',\lambda') \in P'$ such that $\lambda' > \lambda^t$ and $h'_i = \lambda',\ \forall i \in B'$. Hence, we have
\begin{align*}
    \sum_{i \in B'} \tilde a^c_{i}h'_i > \sum_{i \in B'} \tilde a^c_{i}\lambda^t = \sum_{i \in B'} a^c_{i}\tilde{h}_i = \tilde b^c
\end{align*}
and thus $(h', \lambda') \notin Q'$. But, this contradicts the fact that $(\by',h',\lambda') \in P'$.
\end{proof}

We are now ready to prove that the \emph{constrained serial rule} algorithm finds a constrained ordinally efficient assignment.

\begin{theorem}
For any preference profile $\bsucceq$ and constraint set $\Delta \cC(\bsucceq)$, the outcome of Algorithm \ref{alg:genps} is constrained ordinally efficient.
\end{theorem}
\begin{proof}
 Let $\bx$ be the solution of the algorithm for any preference profile $\bsucceq$ and constraint set $\Delta \cC(\bsucceq)$. In order to prove that $\bx$ is constrained ordinally efficient it suffices to show that if there exists $\bx' \in \Delta \cC(\bsucceq)$ such that $\bx'_i sd(\succeq_i) \bx_i$ for all $i \in N$, then $x'_{i}(T_i(\ell)) = x_{i}(T_i(\ell))$, for all $i \in N$, for all $\ell \in \{1,2,...,E(\succeq_i)\}$. We prove this using contradiction. 

For any round $t$, let $(x^t, h^t, \lambda^t)$ denote the optimal solution to $LP(B^t, F^t, (\ell_i^t)_{i \in N})$. For any agent $i$ in the bottleneck set $B^t$, the mechanism fixes the cumulative allocation received by agent $i$ for her top $\ell_i^t$ indifference classes. Hence we have $x_i(T_i(\ell_i^t)) = x^t_i(T_i(\ell_i^t)) = \lambda^t$. Towards a contradiction, let $t$ be the first step in the algorithm such that $x'_{j} (T_j(\ell^t_j)) \neq x^t_{j}(T_j(\ell^t_j))$ for some agent $j \in B^t$. Since $\bx'_i sd(\succeq_i) \bx_i$ for all $i \in N$, it must be that $x'_{j}(T_j(\ell^t_j)) >  x_{j}(T_j(\ell^t_j)) = x^t_{j}(T_j(\ell^t_j)) = \lambda^t$. Also, we have $x'_{i}(T_i(\ell^t_i)) \geq  x_{i}(T_i(\ell^t_i)) = \lambda^t$ for all agents $i \in B^t$ with $i \neq j$. Further, since $\bx'$ is a feasible random assignment, we have $\bx' \in \Delta \cD$ and $A\bx' \leq \boldsymbol{b}$. Lastly, since round $t$ is the first time $\bx'$ differs from $\bx$, $x'_i(T_i(\ell)) = x_i(T_i(\ell)) \geq \gamma$, $\forall (i,\ell,\gamma) \in F^t$. This is in direct contradiction to Lemma \ref{lem:bottleneck}. Therefore, $\bx$ is constrained ordinally efficient.
\end{proof}

In the presence of arbitrary constraints, the existence of envy maybe inevitable between any pair of agents if one agent is more constrained than the other. However, as the next theorem shows we can guarantee envy-freeness among any pair of agents of identical type.

\begin{theorem}
At any preference profile $\bsucceq$, the constrained serial rule guarantees envy-freeness among agents of the same type.
\end{theorem}

\begin{proof}
Let $\bx$ denote the outcome of the algorithm. Consider any pair of agents $i$ and $j$ that are of the same type. Suppose for contradiction that agent $i$ envies $j$, i.e. there exists an indifference class $\ell < E(\succeq_i)$ such that $x_{i}(T_i(\ell)) < x_{j}(T_i(\ell))$. Let $t$ denote the step of the algorithm when $\ell_i^t = \ell$ and agent $i$ is in the bottleneck set $B^t$. Let $(\bx^t, \bh^t, \lambda^t)$ denote the optimal solution to $LP(B^t, F^t, (\ell_i^t)_{i \in N})$. Since agent $i$ is in the bottleneck set $B^t$, the mechanism fixes the cumulative allocation received by agent $i$ for her top $\ell$ indifference classes. Hence we have $x_i(T_i(\ell)) = x^t_i(T_i(\ell)) = \lambda^t$.

Let's consider two cases:

\emph{Case 1: } Suppose agent $j \in B^t$. We have  $x_j(T_j(\ell^t_j))= x_j^t(T_j(\ell^t_j)) = \lambda^t$. However, as $x_j(T_i(\ell)) > \lambda^t$, there must exist some object $o \in T_i(\ell) \setminus T_j(\ell_j^{t})$ such that $x_{j,o} > 0$. 

\emph{Case 2: } Suppose agent $j \notin B^t$. By construction, for any triple $(j, \ell', \gamma) \in F^t$, we have $\ell' \leq \ell_j^{t-1}$ and  $x_j(T_j(\ell')) = \gamma \leq \lambda^t$. Thus, there must exist some object $o \in T_i(\ell) \setminus T_j(\ell_j^{t-1})$ such that $x_{j,o} > 0$.

In either case, since $x_i(T_i(\ell)) = \lambda^t < 1$, there exists some object $p \notin T_i(\ell)$ where $x_{i,p} > 0$. Let $0 < \eps < \min\{x_{j,o}, x_{i,p}\}$ be some fixed constant.
We can now define a new outcome $\by$ as follows. Let
$y_{i',o'} = x_{i',o'}$ for all objects $o' \in O$ and agents $i' \notin \{i,j\}$. For agents $i' \in \{i,j\}$, let $y_{i',o'} = x_{i',o'}$ for all objects $o' \notin \{o,p\}$. Let $y_{i,o} = x_{i,o} + \eps$, 
    $y_{i,p} = x_{i,p} - \eps$, and 
    $y_{j,o} = x_{j,o} - \eps$,
    $y_{j,p} = x_{j,p} + \eps$. Since agents $i$ and $j$ are of the same type and we have $y_{i,o'} + y_{j,o'} = x_{i,o'} + x_{j,o'}$ for all objects $o' \in O$, we must have $A \by = A \bx \leq \bb$. Further, by construction we have $\by \in \Delta \mathcal{D}$, i.e. $\by$ is a feasible outcome.
    In addition, by our choice of objects $o$ and $p$, we have $y_k(T_k(\ell)) \geq x_k(T_k(\ell)) \geq \gamma$ for all $(k,\ell,\gamma) \in \mathcal{F}^t$ and also $y_k(T_k(\ell_k^t)) \geq x_k(T_k(\ell_k^t))$ for all $k \in B^t$. However, since $y_i(T_i(\ell_i^t)) > x_i(T_i(\ell_i^t)) \geq \lambda^t$, this contradicts Lemma \ref{lem:bottleneck}.
\end{proof}

\subsection{Computational Complexity}
\label{sec:complexity}
As shown in Proposition \ref{prop:terminates}, the mechanism terminates in at most $n \rho$ rounds where $n$ and $\rho$ denote the number of distinct agents and objects respectively. In each round $t$, the algorithm solves one instance of the linear program to compute the value of $\lambda^t$. Further, the algorithm needs to find a set of bottleneck agents $B^t$. We now show that $B^t$ can be found in polynomial time by solving a sequence of at most $n$ linear programs.

We recall that $B^t$ is defined as any minimal set of agents such that the objective value of $LP(B^t, F^t, (\ell_i^t)_{i \in N})$ equals $\lambda^t$. Algorithm \ref{alg:bottleneck} provides a simple iterative procedure to find such a minimal set. We first initialize $B^t$ to be the set of all agents. In each step, the algorithm considers removing an agent $i$ from $B^t$. If removing such an agent allows the linear program to obtain a higher objective value, then clearly agent $i$ must belong to the bottleneck set. On the other hand, if the linear program obtains an objective value of only $\lambda^t$, then agent $i$ can be safely removed from consideration since by definition $B^t \setminus \{i\}$ is a smaller candidate set. 

\begin{algorithm}[htb]
\KwIn{$F^t, (\ell_i^t)_{i \in N}, \lambda^t$ as defined in Algorithm \ref{alg:genps}}
\KwResult{$B^t$: Set of bottleneck agents}
$B^t \leftarrow N$\;
\For{$i = 1, 2, \ldots, n$}
{
$\lambda \leftarrow $ objective value of $LP(B^t \setminus \{i\}, F^t, (\ell_i^t)_{i \in N})$\;
\If{$\lambda == \lambda^t$}
{$B^t \leftarrow B^t \setminus \{i\}$}
}
Return $B^t$
\caption{Procedure to find the bottleneck set}
\label{alg:bottleneck}
\end{algorithm}

Algorithm \ref{alg:bottleneck} terminates in at most $n$ iterations and thus in total each round of the constrained serial rule requires solving at most $(n+1)$ linear programs. Algorithm \ref{alg:genps} can thus be executed in time that is polynomial in the size of the constraints, number of agents and objects.

\subsection{Implementability}
\label{sec:implementability}
Randomization in object allocation mechanisms is often used as a tool to incorporate fairness from an ex-ante perspective. The outcome of the random assignment mechanism is treated as a probability distribution over deterministic outcomes and an outcome drawn from this distribution is what gets implemented in practice. By the Birkhoff-von Neumann theorem, it is well known that every bistochastic random assignment can be implemented efficiently as a lottery over feasible deterministic assignments, i.e., a deterministic assignment where every agent is assigned one object and each object is assigned to one agent.
However, in the presence of arbitrary constraints on the random assignment, such a decomposition into a lottery over deterministic assignments that satisfy those constraints may not exist. The following example illustrates such a situation.

\begin{example}
Consider a simple example with one agent, $N = \{1\}$, and three objects $O = \{a,b,c\}$. In the presence of constraints $x_{1,a} + x_{1,b} \leq 2/3$, $x_{1,b} + x_{1,c} \leq 2/3$, and $x_{1,a} + x_{1,c} \leq 2/3$, a potential feasible solution to the random assignment problem is to set $x_{1,a} = x_{1,b} = x_{1,c} = 1/3$. However, it can be readily seen that there exists no deterministic assignment $X$ that satisfies all three constraints and still obtains $X_{1,a} + X_{1,b} + X_{1,c} = 1$.
\end{example}

In many practical applications of constrained object allocation, one can show that any random assignment satisfying these constraints can be represented as a lottery over deterministic assignments that are approximately feasible. For example, for the combinatorial assignment problem with limited complementarities, \citet{nguyen2016assignment} show that one can always decompose a feasible random assignment into a lottery over deterministic assignments where the capacity of each object is violated by at most additive $k$ where $k$ denotes the size of the largest bundle. We discuss such implementation details where applicable in the specific applications in Section \ref{sec:applications}.

On the other hand, our model of imposing constraints on the random assignment solution generalizes the approach of imposing constraints on the ex-post deterministic outcomes. As shown by \citet{balbuzanov2019constrained}, any set of arbitary constraints on the ex-post outcomes can be represented by a set of linear inequalities on the random assignment. While such a reduction always exists, we note that it may not be computationally efficient.

\section{Applications}
\label{sec:applications}
In this section, we discuss how the constrained serial rule can be applied for several concrete applications of constrained object allocation. 

\subsection{Unconstrained Object Allocation}
We can apply our model to the unconstrained object assignment problem by simply setting $\Delta \cC(\bsucceq) = \Delta \mathcal{D}$ for all $\bsucceq \in \bbR^n$. In this case, the random assignment output by the \emph{constrained serial rule} coincides with that given by the extended probabilistic serial algorithm of \cite{katta2006solution}. 

We first briefly discuss the extended probabilistic serial algorithm. At every round $t$, the extended probabilistic serial algorithm constructs a flow network where agents point to their most preferred objects among the set of objects available at that round. Using the parametric max-flow algorithm, the algorithm then identifies a bottleneck set of agents $X$ to be those that satisfy:
\begin{align*}
    X = \argmin_{Y \subseteq N}\frac{|\Gamma(Y)|}{|Y|}
\end{align*}
where $\Gamma(Y)$ denotes the set of objects that are most preferred by atleast one agent in $Y$. 
We note that this set of agents $X$ is precisely the most constrained set in this round, i.e., an agent $i \in X$ can get exactly $\frac{|\Gamma(X)|}{|X|}$ (and not more) total probability share from her preferred objects. In the constrained serial rule algorithm, since there are not other constraints restricting the random assignments, this same set $X$ of agents will be chosen as the bottleneck set $B^t$.

\subsection{Bi-hierarchical Constraints}

\citet{budish2013designing} considers the object allocation problem with general quota constraints and identified a large class of constraints called `bi-hierarchical constraints' that are universally implementable, i.e., any random assignment satisfying these constraints can be implemented as a lottery over deterministic assignments that satisfy the same constraints.

A constraint in their setup is of the form $\underline{q}_S \leq \sum_{(i,o) \in S} \overline{x}_{i,o} \leq \overline{q}_S$, where $S \subseteq N \times O$ is a set of agent-object pairs, $\overline{\bx}$ is a deterministic assignment, and $\underline{q}_S,\overline{q}_S$ are both integers. A constraint structure $\mathcal{H}={(S,\boldsymbol{q}_S)}$ comprises of collection of such constraints. An additional requirement on $\mathcal{H}$ is that it must include all singleton sets. 
A constraint structure $\mathcal{H}$ is a hierarchy if for every $S,S' \in \mathcal{H}$, either $S \subseteq S'$ or $S' \subseteq S$ or $S \cap S' = \emptyset$. Finally, $\mathcal{H}$ is a bi-hierarchy if there exist hierarchies $\mathcal{H}_1, \mathcal{H}_2$ such that $\mathcal{H} = \mathcal{H}_1 \cup \mathcal{H}_2$ and $\mathcal{H}_1 \cap \mathcal{H}_2 = \emptyset$. \citet{budish2013designing} proposed the \emph{Generalized Probabilistic Serial} mechanism for object allocation with bi-hierarchical quota constraints. However, their mechanism only works for upper-bound quotas and assumes that all lower bounds $\underline q_S = 0$.

In contrast, we can directly use the inequalities in the constraint structure $\mathcal{H}$ in our \emph{constrained serial rule} algorithm (even in the presence of indifferences in the preference relations) by defining:

\[\Delta \cC(\bsucceq) = \{\bx \in \Delta \cD \mid \underline{q}_S \leq \sum_{(i,o) \in S} x_{i,o} \leq \overline{q}_S \text{ for all } (S,\boldsymbol{q}_S) \in \mathcal{H}\} \text{ for all } \bsucceq \in \bbR^n\]
Thus our algorithm generalizes the approach of Budish et al.\ even for lower bound quota constraints. The key technical innovation that allows us to do so lies in the fact that our algorithm \emph{looks ahead} in time to ensure that the partial solution obtained at any time leads to a feasible random assignment.

\subsection{Type-Dependent Distributional Constraints}
\cite{ashlagi2020assignment} study type dependent distributional constraints that do not conform to a bi-hierarchical structure. In this setup, every agent $i \in N$ is associated with a type $t_i \in T$ where $T$ denotes a finite set of types. Let $R \subseteq T$ denote an arbitrary set of agent types and let $o \in O$ denote an arbitrary object. A single constraint is of the form $\underline q_{R,o} \leq \sum_{i \in N \mid t_i \in R} x_{i,o} \leq \overline q_{R,o}$, i.e., the mechanism imposes floor and ceiling quotas on the total allocation of all agents belonging to a specific set of agent types at a given object.

As earlier, it can be readily seen that such distributional constraints can be easily represented in our framework. Since the constraints do not conform to a bi-hierarchical structure, the outcome of our mechanism cannot always be implemented as a lottery over feasible deterministic assignments. However, as shown by \citet{ashlagi2020assignment}, any random assignment satisfying these constraints can be decomposed into a distribution over almost feasible deterministic outcomes where every floor and ceiling constraint is violated by at most $|T|$.

\subsection{Explicit Ex-post Constraints}
\cite{balbuzanov2019constrained} considers the problem of random object assignment when we are given an explicit list of ex-post feasible allocations and the random assignment must be implemented as a lottery over these allocations. For a preference profile $\bsucceq$, let $C(\bsucceq)$ be the set of all permissible deterministic assignments and $\Delta \cC(\bsucceq)$ be the convex hull of the set $\cC$. He shows that for every $C(\bsucceq)$, there exists a minimal set of constraints parameterized by the matrix $A$, with $a^c_{i,o} \geq 0$ for all $(i,o) \in N \times O$ and constraint $c$, and the vector $\boldsymbol{b} \geq \boldsymbol{0}$ such that $\Delta \cC(\bsucceq) = \{\bx \in \Delta \cD \mid A\bx \leq \boldsymbol{b}\}$. He generalizes the probabilistic serial mechanism to incorporate these inequalities for the case when agents have strict preferences. The \emph{constrained serial rule} algorithm generalizes his mechanism to the full preference domain.

\subsection{Combinatorial Assignment}
Another class of problems where our mechanism can be applied to is the problem of allocating bundles of indivisible objects to agents when preferences exhibit complementarities (\cite{budish2011combinatorial,budish2012multi,nguyen2016assignment}). Formally, let $G$ be an underlying set of objects, where each object $g \in G$ is supplied in $q_g$ copies. A bundle of objects can be represented by a vector in $\mathbb{N}_{+}^{|G|}$, where the $t$th co-ordinate of this vector corresponds to the number of copies of the object $t$ and $\mathbb{N}_{+} = \mathbb{N} \cup \{0\}$. Let $O=\{o \in \mathbb{N}_{+}^{|G|} \mid \sum_{g \in G} o_{g} \leq k\}$ now be the set of all bundles of size at most $k$. We assume that each bundle is available in a single copy. Each agent $i \in N$ is interested in consuming one bundle from the set $O$ and has a complete and transitive preference $\succeq_i$ on the set $O$. A common application that fits in this class of problems is the course allocation problem. Every student is to be assigned a schedule of at most $k$ courses, where each course $g$ has a finite number of seats $q_g$. 

The set of feasible random allocations can be described by the following set of constraints that enforce that the total amount allocated of any object $g \in G$ is at most its supply $q_g$. At every preference profile $\bsucceq$, we have:
\begin{align*}
    \Delta \cC(\bsucceq) = \Delta \cC = \{ x \in \Delta \cD \mid \sum_{i \in N}\sum_{o \in O} o_g \cdot x_{i,o} \leq q_g, \quad \forall g \in G\} 
\end{align*}

From Definition \ref{def:sametype}, it is easy to see that all agents under these feasibility constraints are of the same type. Therefore, the \emph{constrained serial rule} guarantees constrained ordinally efficient and envy-free outcomes. Further, as discussed by \citet{nguyen2016assignment}, any outcome of the mechanism can be implemented as a lottery over deterministic assignments that violate the supply constraints by at most $k-1$.

\newpage

\bibliography{bibliography}

\end{document}